%% file: main.tex
\documentclass[sigconf]{acmart}
\usepackage{amssymb}
\usepackage{amsmath}
\usepackage{makecell}
\usepackage{balance}
\usepackage{float}
\usepackage{mathtools}
\usepackage{setspace}
\usepackage{array}
\usepackage{paralist}

\usepackage{stfloats}

\renewcommand\footnotetextcopyrightpermission[1]{} % removes footnote with conference information in first column
\author{Yishai Oltchik}
\email{yishai.oltchik@inf.ethz.ch}
\affiliation{ETH Zurich}    
\author{Oded Schwartz}
\email{odedsc@cs.huji.ac.il}
\affiliation{Hebrew University of Jerusalem}

\title{Network Partitioning and Avoidable Contention}

\setcopyright{rightsretained}

\date{}

\begin{document}

\input{abstract}
\maketitle
\thispagestyle{empty}
\input{introduction}
\input{preliminaries}

\input{analysis}

\input{experiments}
\input{discussion}
\input{acknowledgments}
\balance
\bibliographystyle{plain}
\bibliography{references}
\clearpage
\input{appendix}

\end{document}

%% file: abstract.tex
\begin{abstract}\small\baselineskip=9pt
Network contention frequently dominates the run time of parallel algorithms and limits scaling performance.
Most previous studies mitigate or eliminate contention by utilizing one of several approaches: communication-minimizing algorithms; hotspot-avoiding routing schemes; topology-aware task mapping; or improving global network properties, such as bisection bandwidth, edge-expansion, partitioning, and network diameter.
In practice, parallel jobs often use only a fraction of a host system.
How do processor allocation policies affect contention within a partition?
We utilize edge-isoperimetric analysis of network graphs to determine whether a network partition has optimal internal bisection.
Increasing the bisection allows a more efficient use of the network resources, decreasing or completely eliminating the link contention.
We first study torus networks and characterize partition geometries that maximize internal bisection bandwidth.
We examine the allocation policies of Mira and JUQUEEN, the two largest publicly-accessible Blue~Gene/Q torus-based supercomputers.
Our analysis demonstrates that the bisection bandwidth of their current partitions can often be improved by changing the partitions' geometries.
These can yield up to a $\times 2$~speedup for contention-bound workloads.
Benchmarking experiments validate the predictions.
Our analysis applies to allocation policies of other networks.
\end{abstract}

%% file: introduction.tex
\section{Introduction}
Network contention frequently dominates the run time of parallel algorithms and limits scaling performance~\cite{ballard2014communication}.
Most previous studies mitigate or eliminate contention by utilizing one of several approaches: communication-minimizing algorithms (cf. ~\cite{solomonik2011communication,caps,demmel2012communication}); hotspot-avoiding routing schemes (cf. ~\cite{singh2005load}); topology-aware task mapping (cf. ~\cite{bhatele2014optimizing}); or improving global network properties such as bisection bandwidth, edge-expansion (cf. ~\cite{besta2014slim, valadarsky2015xpander}), partitioning (cf. ~\cite{jain2017partitioning}), and network diameter (cf. ~\cite{kim2008technology}).
Parallel jobs running on a supercomputer or a cloud platform often do not utilize the entire machine at once.
Rather, the job is assigned a subset of the system's compute nodes and associated resources for its exclusive use\footnote{Some cloud platforms allow `multi-tenancy', in which case exclusivity is not guaranteed. This adds further challenge which we do not address in this paper.}.
Optimizing the internal bisection bandwidth of allocated partitions can decrease or completely eliminate the link contention of a parallel computation, improving overall performance for contention-bound workloads.

\paragraph*{Our contribution}
Using isoperimetric analysis, we study torus networks and characterize partition geometries that maximize internal bisection bandwidth.
Our analysis utilizes a novel generalization of Bollob{\'a}s and Leader's bounds on the edge-isoperimetric problem on torus graphs~\cite{BollobasLeader91}.
A solution was known for tori with dimensions of equal size, whereas our new bound applies to torus graphs with arbitrary dimension sizes.
This is useful, as the vast majority of torus networks with $3$ dimensions or more have unequal dimensions.
We apply isoperimetric analysis to compute node partition allocations allowed by the allocation of Mira and JUQUEEN, the two largest publicly-accessible Blue~Gene/Q torus-based supercomputers.
Our analysis demonstrates that the bisection bandwidth of their current partitions can often be improved by changing the partitions' geometries, yielding up to a $\times 2$~speedup for contention-bound workloads.
Benchmarking experiments on both systems validate the predictions.
We show an impact of 10\% speedup for fast matrix multiplication.
We also applied the analysis to the Blue~Gene/Q machine Sequoia, but no experiments were performed as it is no longer available for scientific research.
Lastly, we discuss network configurations for hypothetical Blue~Gene/Q systems which, despite having fewer compute and network resources, may perform better. They are predicted to improve upon the network performance of JUQUEEN by increasing the bisection bandwidth of partitions.
This work focuses first on the Blue~Gene/Q supercomputer series, but the application of our method to other networks topologies such as hypercubes, Dragonfly, Slim Fly, and HyperX, is also described in detail.
\input{related}
\input{organization}

%% file: related.tex
\paragraph*{Related Work}
Bisection bandwidth is a standard metric for network performance.
An in-depth description of the Blue~Gene/Q topology appears in~\cite{underthehood}, which also includes an analysis of its bisection bandwidth, and outlines traffic patterns that are challenging for the network to efficiently route.
However, they do not discuss the bisection bandwidth of network partitions.
Finding worst-case traffic patterns for an arbitrary network topology can be non-trivial in the general case.
A method for generating ``near-worst-case" traffic patterns is shown in~\cite{jyothi2016measuring}.
The edge-isoperimetric problem (see definition in Section~\ref{section:isoperimetric}) is a well-known problem in combinatorics, and general solutions have been shown for several graphs that either directly correspond or are very similar to network topologies used in practice.
These include:
hypercubes~\cite{harper1964optimal}, 
cubic tori~\cite{BollobasLeader91},
Cartesian products of cliques~\cite{lindsey1964assignment},
and 2-dimensional mesh grids~\cite{ahlswede1995edge}.
The edge-isoperimetric problem provides a tight bound on the bandwidth between two arbitrary sides of the network.
Since both sides may have the same size, the edge-isoperimetric problem generalizes the problem of determining the bisection bandwidth of a graph.
Closely related to the edge-isoperimetric problem is the notion of small-set expansion in graphs (see Section~\ref{section:prelim}). Indeed, if a graph $G$ is $d$-regular then the two problems are essentially equivalent.
Spectral methods that can be used to approximate the small-set expansion of arbitrary graphs are described in~\cite{lee2014multiway}.
The small set expansion of a network graph is used in~\cite{contention} to derive lower bounds on the contention costs, and potentially determine when a given parallel algorithm on a given system is inevitably asymptotically contention-bound.

%% file: organization.tex
\paragraph*{Paper Organization}
In Section~\ref{section:prelim} we provide preliminaries on the edge-isoperimetric problem on torus graphs, and on the IBM Blue~Gene/Q architecture.
In Section~\ref{section:analysis} we present isoperimetric analysis of general tori and apply it to the partitions of Mira and JUQUEEN, concluding with alternative, improved partitions.
In Section~\ref{section:evaluation} we perform experiments on Mira and JUQUEEN, and discuss the results.
In Section~\ref{section:discussion} we present implications of our analysis on networks design, discuss the applicability of our methods to other network topologies including ToFu, Dragonfly, Fat-Tree, and HyperX, and outline future work.

%% file: preliminaries.tex
\section{Preliminaries}\label{section:prelim}
A main application of our method is improving allocation policies of the torus-based Blue~Gene/Q systems.
We begin by defining torus graphs and our primary analysis tool of the edge-isoperimetric problem.

\subsection*{Torus graphs}
Let $D$ and $a_1, \ldots, a_D$ be integers, and let $G = \left(V,E\right)$ be a graph.
If $V = \left[a_1\right] \times \ldots \times \left[a_D\right]$, and every two vertices $u = \left(u_1, \ldots, u_D\right), v = \left(v_1, \ldots, v_D\right)$ are adjacent if and only if $\exists k$ such that $u_k = v_k \pm 1 \bmod a_k$ and $\forall j \neq k, u_j = v_j$, then $G$ is said to be a \emph{$D$-torus} (also, $D$-dimensional torus).
If $a_1 = \ldots = a_D$ then $G$ is said to be a \emph{cubic} torus.

Let $G=\left(V,E\right)$ be a graph, and let $A,B\subset V$.
Then, the \emph{perimeter} of $A$ is $E\left(A, \overline{A}\right) = \left\lbrace u,v \mid u \in A, v \notin A\right\rbrace$ and the \emph{interior} of $A$ is $E\left(A, A\right) = \left\lbrace u,v \mid u \in A, v \in A\right\rbrace$.
For any $k$-regular graph, the following equation holds:
\begin{equation}\label{eqn:graph-regular-interior-perimeter}
\forall A \subseteq V, k \left|A\right| = 2 \left|E\left(A,A\right)\right| + \left|E\left(A,\overline{A}\right)\right|
\end{equation}
\input{isoperimetric}

\paragraph*{Small Set Expansion}
The small-set expansion of a graph $G = \left(V,E \right)$, denoted $h_t\left(G\right)$, is defined:
$$h_t \left(G\right) = \min_{\substack{A \subset V \\ \left|A\right| \leq t}} 
\frac{\left|E\left(A, \overline{A}\right)\right|}
{\left|E\left(A, A\right)\right| + \left|E\left(A, \overline{A}\right)\right|}$$
Small-set expansion can be used to test whether a given network will be inevitably asymptotically contention-bound when executing a parallel algorithm with known per-processor communication costs~\cite{contention}.
Since the small-set expansion is attained by the bisection for all networks and partitions considered in this work, it will suffice for us to consider only the bisection bandwidth.
\input{bgq}

%% file: isoperimetric.tex
\paragraph*{Edge-isoperimetric problem on torus graphs}\label{section:isoperimetric}
The edge-isoperimetric problem is defined as follows: given a graph $G=\left(V, E\right)$ and some integer $t \le \frac{\left|V\right|}{2}$, find $S \subset V$ with $\left|S\right| = t$ of minimal perimeter size. That is, find $S$ such that:
$$\left|E\left(S, \bar{S}\right)\right| = \min_{\substack{A \subset V \\ \left|A\right| = t}}  {\left\{\left|E\left(A, \bar{A}\right)\right|\right\} }$$
Such a set $S$ is said to be \emph{isoperimetric}.
Note that by Equation~\ref{eqn:graph-regular-interior-perimeter}, for $k$-regular graphs, minimizing the perimeter is equivalent to maximizing the interior.

If $G$ is a cubic torus, then the following bound of Bollob{\'a}s and Leader~\cite{BollobasLeader91} applies:
\begin{theorem}[Edge-isoperimetric ineq. for cubic tori]\label{thm:cubic-torus}
Let $G=\left(V,E\right)$ be a cubic $D$-dimensional torus such that $V = \left[n\right]^D$, and let $t \le \frac{n^D}{2}$. Then $\forall S\subset V$ with $\left\vert S \right\vert = t$:
\begin{equation}\label{cubic-torus-sse}
\left|E\left(S,\overline{S}\right)\right| \ge 
\min_{r\in\left\{ 0,\ldots,D-1\right\} }
2\left(D-r\right)\cdot n^{\frac{r}{D-r}}\cdot t^{\frac{D-r-1}{D-r}}
\end{equation}
\end{theorem}
For $r$ such that $\left(\frac{t}{n^r}\right)^{\frac{1}{D-r}}$ is an integer, we define $S'\subset V$ such that:
$$S'= \left[n\right]^r \times \left[\left(\frac{t}{n^r}\right)^{\frac{1}{D-r}}\right]^{D-r}$$
In this case, $S'$ is a $D$-dimensional cuboid with $r$ dimensions of length $n$, and $D-r$ dimensions of length $\left(\frac{t}{n^r}\right)^{\frac{1}{D-r}}$.
Each vertex contributes $2\left(D-r\right)$ edges to the cut.
Since $\left\vert S' \right\vert = t$, a simple counting argument leads to: $$\left\vert E(S', \bar{S'})\right\vert = 2r \cdot t^{\frac{D-r-1}{D-r}} \cdot n^{\frac{r}{D-r}}$$
Therefore the bound presented in Theorem~\ref{thm:cubic-torus} is tight for certain values of $t$.

%% file: bgq.tex
\subsection*{Blue~Gene/Q Systems}\label{section:bgq}
IBM Blue~Gene/Q systems~\cite{chen2012ibm} have $5$D torus network topologies where the size of at least one dimension is exactly $2$.
The bisection bandwidth of a Blue~Gene/Q system is $2 \cdot \frac{N}{L} \cdot B$, where $N$ is the number of nodes, $L$ is size of the longest dimension, and $B$ is the capacity of a single bidirectional link~\cite{underthehood}.
A midplane in the Blue~Gene/Q topology is a physical arrangement of $512$ compute nodes, internally connected by a $5$D torus network with dimensions $4 \times 4 \times 4 \times 4 \times 2$.
The last dimension, of length $2$, is internal to the midplane.
A physical rack in a Blue Gene/Q system consists of two midplanes.
$13$ Blue~Gene/Q systems appear in the November 2017 list of top 500 supercomputers~\cite{TOP500-nov17}.
The network is physically constructed in such a way that partitions may have wrap-around links in a given dimension even when they do not fully cover that dimension in the entire network.
All partitions discussed in this work are $4$-dimensional sub-tori where some dimensions may have size $1$.
There are no published limits to the maximal size of a Blue~Gene/Q system or to the lengths of any\footnote{Except the $5$th dimension, which has size $2$ and is internal to each midplane.} dimension~\cite{underthehood}.

To simplify notation, we always present the dimensions of a torus network and its partitions in sorted order by length.
This canonical representation treats partitions whose geometries are identical up to rotations as one.
Therefore, a machine with network size $2 \times 2 \times 2 \times 1 \times 1$ fits $4$ partitions with geometry $2 \times 1 \times 1 \times 1 \times 1$.

Outside of jobs which require an exceptionally small amount of compute nodes, all partitions in Blue~Gene/Q systems are defined by cuboids (Cartesian products of chains and cycles) consisting of whole midplanes.
We therefore represent the Blue~Gene/Q network and its partitions as $4$-dimensional tori of midplanes.
For example, consider a $6$-midplane system of dimensions $3 \times 2 \times 1 \times 1$. In terms of compute nodes, this system has $3072$ compute nodes and network size $12 \times 8 \times 4 \times 4 \times 2$.
The best possible $1536$-compute node partition of this system has dimensions $12 \times 4 \times 4 \times 4 \times 2$ and $256$ links in its bisection.
An alternate partition with dimensions $8 \times 6 \times 4 \times 4 \times 2$ would have the same node count, but a greater bisection of $384$. However, since its largest dimension consists of $1.5$ midplanes it is not supported by the Blue~Gene/Q topology.
Such a partition could be constructed by over-provisioning an additional midplane and defining a partition with dimensions $8 \times 8 \times 4 \times 4 \times 2$.
Our benchmarks and applications all use message-passing communication with MPI, which allows the individual processes (often referred to as \emph{ranks}) in the computations to communicate directly with each other.
Unless explicitly stated otherwise, each compute node is assigned only one MPI rank.
This allows an improved bisection bandwidth of $512$ links, but at the cost of requiring additional compute nodes. We next introduce the Blue~Gene/Q systems Mira and JUQUEEN.
\paragraph{Mira}\label{sub:mira}
Installed at Argonne National Laboratory~\cite{mira}, it is the largest Blue~Gene/Q system accessible for scientific research. Mira is ranked $24$th in the July 2019 Top 500 supercomputers~\cite{TOP500-nov17}.
It has $49152$ compute nodes, with network size $16 \times 16 \times 12 \times 8 \times 2$, or $4 \times 4 \times 3 \times 2$ midplanes.

\paragraph{JUQUEEN}\label{sub:juqueen}
Installed at J{\"u}lich Supercomputing Centre, JUQUEEN is the second-largest Blue~Gene/Q system accessible for scientific research. It was\footnote{JUQUEEN was since dismantled and does not appear in later lists.} ranked $22$nd in the November 2017 Top 500 supercomputers~\cite{TOP500-nov17}.
JUQUEEN has $28672$ compute nodes, and network size $28 \times 8 \times 8 \times 8 \times 2$, or $7 \times 2 \times 2 \times 2$ midplanes.

%% file: analysis.tex
\section{Theoretical Analysis}\label{section:analysis}
Using isoperimetric analysis, we identify allocation policies that are not optimal; namely, we point to partitions with sub-optimal internal bisection bandwidth.
Whenever such partitions exist, we find partition geometries with optimal bisection bandwidth that are likely to reduce link contention.

\subsection{The Edge-Isoperimetric Problem}
We obtain a novel generalization of Theorem~\ref{thm:cubic-torus} to arbitrary torus graphs.
We show that the bound is optimal for cuboid subsets, and conjecture that it is optimal for arbitrary subsets as well.
\begin{theorem}[Edge-isoperimetric ineq. for tori]\label{thm:general-torus}
    Let $G=\left(V,E\right)$ be a $D$-dimensional torus with $V = \left[a_1\right] \times \left[a_2\right] \times \ldots \times \left[a_D\right]$, and $t \le \frac{\left|V\right|}{2}$. Suppose, without loss of generality, that $a_1 \ge a_2 \ge \ldots \ge a_D$. Then, for any cuboid $S\subset V, \left\vert S \right\vert = t$:
    \begin{equation}\label{general-torus-sse}
    \left|E\left(S,\overline{S}\right)\right|\ge \min_{r\in\left\{ 0,\ldots,D-1\right\} }2\left(D-r\right)
    \left(\prod_{i=0}^{r-1} a_{D-i}\right)^{\frac{1}{D-r}}
    t^{\frac{D-r-1}{D-r}}
    \end{equation}
Like Theorem~\ref{thm:cubic-torus}, our bound can be attained in some cases.
Let $k = \prod_{i=0}^{r-1} a_{D-i}$. If $\exists r$ such that $\left(\frac{t}{k}\right)^{\frac{1}{D-r}}$ is an integer, define the cuboid
$S_r = \left[\left(\frac{t}{k}\right)^\frac{1}{D-r}\right]^{D-r} \times \left[a_{D-r+1} \right] \times \ldots \times \left[a_D \right]$.
\end{theorem}
Our proof strategy for Theorem~\ref{thm:general-torus} is as follows: in Lemma~\ref{claim:attain-cut} we show an explicit construction of a class of cuboid sets in general torus graphs whose cut size matches Equation~\ref{general-torus-sse}.
In Lemma~\ref{claim:subsume-minimal-dims} we show these sets are isoperimetric, thereby completing the proof.

\begin{lemma}\label{claim:attain-cut}
Let $G = \left(V,E\right)$ be a $D$-dimensional torus with $V = \left[a_1\right] \times \ldots \times \left[a_D\right]$, and let $t,k,r'$ be integers such that $t \leq \frac{\left|V\right|}{2}$ and $\left(\frac{t}{k}\right)^\frac{1}{D-r'}$ is an integer.
Define $S_{r'}$ as in Theorem~\ref{thm:general-torus}, with $\arg\min = r'$.
Then, $S_{r'}$ maintains:
$$\left|E\left(S_{r'},\overline{S_{r'}}\right)\right| = 2\left(D-r'\right) \left(\prod_{i=0}^{r'-1} a_{D-i}\right)^{\frac{1}{D-r'}} t^{\frac{D-r'-1}{D-r'}}$$
\end{lemma}
\begin{proof}
If $a_1 = \ldots = a_D = 2$, then by Harper~\cite{harper1964optimal}, \(S_{r'}\) is an isoperimetric set maintaining the desired cut size, and the lemma follows.
If all dimension lengths are strictly greater than \(2\), then we show the cut size directly:
We count the edges in $E\left(S_{r'}, \overline{S_{r'}}\right)$ by considering the size of each $\left(D-1\right)$-dimensional face of the cuboid $S_{r'}$.
Faces in dimensions where $S_{r'_i} = a_i$ contribute no edges to the cut.
There are $\frac{t}{\left(\frac{t}{k}\right)^\frac{1}{D-r}} = \left(\prod_{i=0}^{r-1} a_{D-i}\right)^{\frac{1}{D-r}} t^{\frac{D-r-1}{D-r}}$ vertices on each remaining face.
Then: \[\left|E\left(S_{r'},\overline{S_{r'}}\right)\right| = 2\left(D-r'\right) \left(\prod_{i=0}^{r'-1} a_{D-i}\right)^{\frac{1}{D-r'}} t^{\frac{D-r'-1}{D-r'}}\]
$S_{r'}$ is similar to $S'$ as defined in Equation~\ref{cubic-torus-sse}, but instead of having $r'$ dimensions of length $n$, the dimension lengths are $a_D, \ldots, a_{D-r'+1}$.
If only some dimensions \(a_{D-k+1}, \ldots, a_{D}\) have lengths \(2\), then we choose \(S_{r'}\) such that they are all covered, and then proceed as before with \(t' = \frac{t}{2^k}\), and the same cut is attained.
\end{proof}

\begin{lemma}\label{claim:subsume-minimal-dims}
    Let $G, t, D, k, r, S_r$ be defined as in Lemma~\ref{claim:attain-cut}.
    Let $A \subset V$ be some cuboid $\left[A_1 \right] \times \ldots \times \left[A_D\right]$ with $\left|A\right| = t$. Suppose there exist exactly $r$ indices $i_1, \ldots, i_r$ that maintain $A_{i_k} = a_{i_k}$. Then, $\left|E\left(S_r,\overline{S_r}\right)\right| \le \left|E\left(A,\overline{A}\right)\right|$.
\end{lemma}
\begin{proof}
    If $A$ can be transformed into $S_r$ by changing the order of equal-sized dimensions (i.e., $A$ is a rotation of $S_r$) then the equality is trivial.
    By assumption, \(A\) and \(S_r\) both fully cover exactly \(r\) dimensions. Then, since \(A\) is not a rotation of \(S_r\), there must exist dimensions $i,j$ such that $A_i < t^\frac{1}{D-r} < A_j$. That is, the projection of \(A\) on the dimensions \(i, j\) is an oblong rectangle and not a square. In this case, the lemma follows directly by applying to \(A\) the same counting argument used in the proof of Lemma~\ref{claim:attain-cut}.
    
    Suppose $\exists j$ such that $A_j = a_j$ and $S_{r_j} \neq a_j$. Since $A$ is not a rotation of $S_r$, and by definition, $S_{r_i} = a_i$ for $a_{D-r}, \ldots, a_D$, then a face of $A$ contributes at least a factor of $\frac{a_j}{a_{D-r}}$ more edges to the perimeter than a face of $S_r$.
    Thus, $\left|E\left(S_r,\overline{S_r}\right)\right| \le \left|E\left(A,\overline{A}\right)\right|$ with equality if and only if $A$ is a rotation of $S_r$.
\end{proof}

A central implication of Theorem~\ref{thm:general-torus} is that for large values of $t$ with $r_{\text{opt}} = D-1$, the size of the perimeter is bounded below by $2\cdot \prod_{i=2}^{D} a_i$.
In particular, the bisection bandwidth is improved the closer $\frac{a_1}{t}$ is to $t^\frac{D-1}{D}$.
This is consistent with the result of~\cite{underthehood} regarding the bisection bandwidth of the Blue~Gene/Q network, and leads us to an easy corollary.
\begin{corollary}
    Let $G = \left(V, E\right)$ be the network graph of a Blue~Gene/Q machine, and let $A$ be a cuboid of midplanes with dimensions $A_1 \times \ldots \times A_4$.
    If $\exists B \subset V$ with dimensions $B_1 \times \ldots \times B_4$ such that $\left|A\right| = \left|B\right|$ and $\frac{B_1}{\left|A\right|} < \frac{A_1}{\left|A\right|}$, then $B$ has strictly greater internal bisection bandwidth than $A$.
\end{corollary}

\subsection{Analysis of Blue~Gene/Q Systems}
We apply Lemma~\ref{claim:attain-cut} to the partitions in Mira and JUQUEEN, and obtain their bisection bandwidths.
Using Lemma~\ref{claim:subsume-minimal-dims}, partition geometries with improved bisection bandwidth are found.

\subsubsection*{Mira}\label{sub:mira-analysis}
Not all cuboids of midplanes are permitted by Mira's scheduler.
There is a predefined list of partitions that may be used (see Table~\ref{tbl-mira-parts-all} in Appendix B).
Where possible, we propose partitions of identical size and greater internal bisection bandwidth (see Table~\ref{tbl-mira-summary} and Figure~\ref{plt:mira-partitions}).
With the assistance of the operators of Mira, we were able to allocate the new partitions for the duration of our experiments.
This allowed us to conduct benchmark comparisons (see Section~\ref{section:evaluation}).
\subsubsection*{JUQUEEN}\label{sub:juqueen-analysis}
Partitions of JUQUEEN's network can be any cuboids of midplanes that fit inside the full network.
Users may request partitions by specifying either their exact geometry in midplanes, or by specifying only the overall size.
For some sizes, partitions both optimal and sub-optimal in terms of internal bisection bandwidth are permissible by the job scheduler.
When only a partition size is specified, inconsistent performance may occur if one part of a user's executions is allocated optimal partitions, and another part is not.

\begin{figure}[t]
	\includegraphics{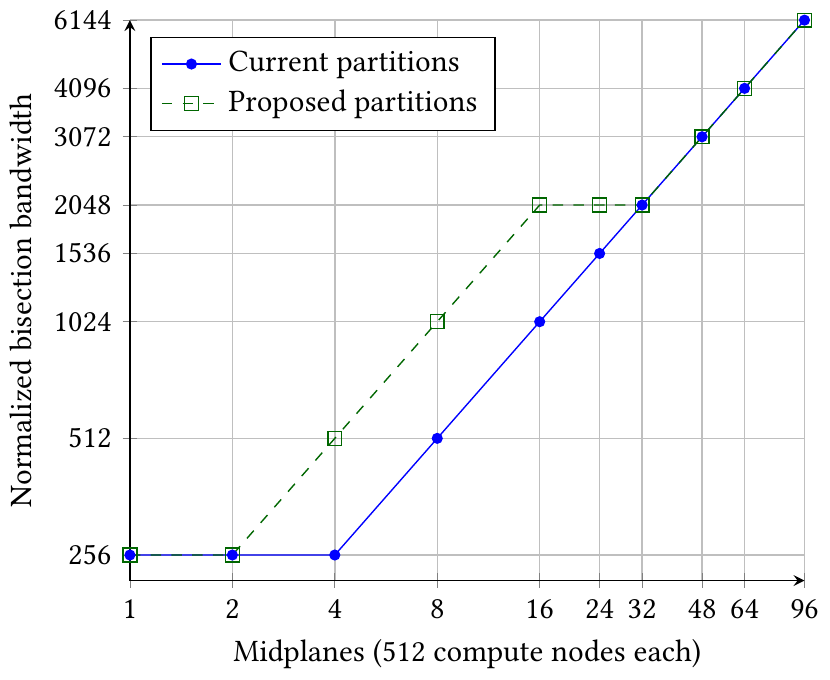}
	\caption{Mira: Normalized bisection bandwidth of currently-defined and proposed partition geometries. Each link contributes $1$ unit of capacity.}\label{plt:mira-partitions}
\end{figure}

\input{tables/mira-partition-summary-table.tex}

\begin{figure}
	\includegraphics{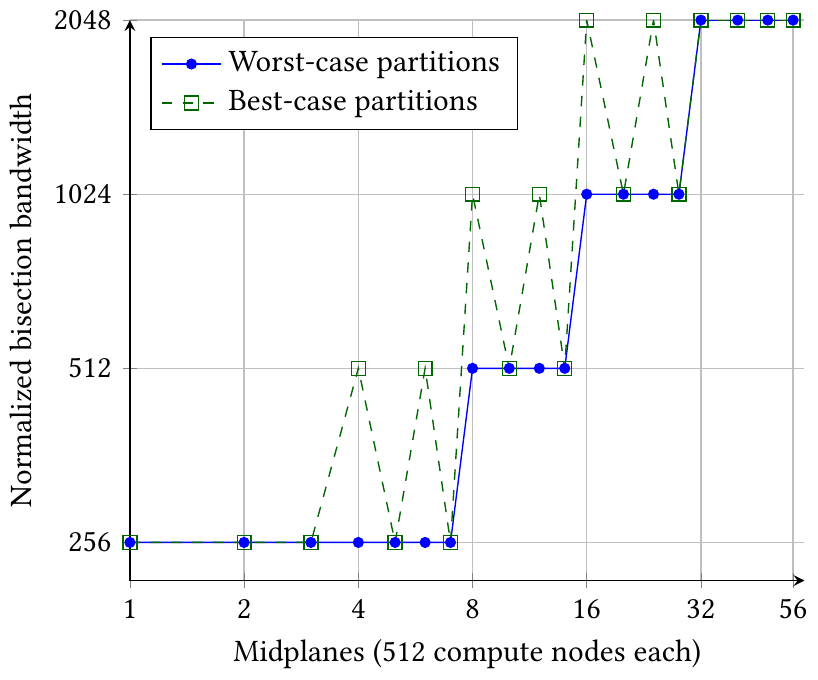}
	\caption{JUQUEEN: Normalized bisection bandwidth of best and worst-case partition geometries. Each link contributes $1$ unit of capacity. The `spiking' drops correspond to partitions whose size requires them to be ring-shaped, and hence have small bisection bandwidth.}\label{juqueen-partitions-inplace}
\end{figure}

\input{tables/juqueen-partitions-table.tex}

%% file: tables/mira-partition-summary-table.tex
\begin{table}[ht!]
\caption{Mira: partial list of normalized bisection bandwidths of current and proposed partitions, showing only rows where the bisection is increased. Full list in Table~\ref{tbl-mira-parts-all}, Appendix B.}\label{tbl-mira-summary}
\resizebox{\columnwidth}{!}{
    \begin{tabular}{|l|c||c|c||c|c|} \hline
    \thead{$P$} &  \thead{Midplanes} & \thead{Current Geometry}         & \thead{BW}  & \thead{Proposed Geometry}            & \thead{Proposed BW}\\ \hline
    $2048$      &  $ 4 $             & $4 \times 1 \times 1 \times 1 $  & $256$       & $2 \times 2 \times 1 \times 1$  & $512$         \\ \hline
    $4096$      &  $ 8 $             & $4 \times 2 \times 1 \times 1 $  & $512$       & $2 \times 2 \times 2 \times 1$  & $1024$        \\ \hline
    $8192$      &  $ 16 $            & $4 \times 4 \times 1 \times 1 $  & $1024$      & $2 \times 2 \times 2 \times 2$  & $2048$        \\ \hline
    $12288$     &  $ 24 $            & $4 \times 3 \times 2 \times 1 $  & $1536$      & $3 \times 2 \times 2 \times 2$  & $2048$        \\ \hline
    \end{tabular}
}
\end{table}

%% file: tables/juqueen-partitions-table.tex
\begin{table}
\caption{JUQUEEN: partial list of normalized bisection bandwidths of optimal and worst-case partitions, showing only rows where best and worst cases differ. Full list in Table~\ref{tab:juqueen-full-allocations}, Appendix B.} \label{tab:juqueen-improved-allocations}
\resizebox{\columnwidth}{!}{%
\begin{tabular}{|l|c||c|c||c|c|} \hline
\thead{$P$} & \thead{Midplanes}           & \thead{Worst Geometry}         & \thead{Worst BW}  &\thead{Best Geometry}             & \thead{Best BW}    \\   \hline
$2048 $     & $ 4 $           & $4 \times 1 \times 1 \times 1$      &  $256$             &$2 \times 2 \times 1 \times 1$    &  $512$  \\   \hline
$3072 $     & $ 6 $           & $6 \times 1 \times 1 \times 1$      &  $256$             &$3 \times 2 \times 1 \times 1$    &  $512$  \\   \hline
$4096 $     & $ 8 $           & $4 \times 2 \times 1 \times 1$      &  $512$             &$2 \times 2 \times 2 \times 1$    &  $1024$ \\   \hline
$6144 $     & $ 12 $           & $6 \times 2 \times 1 \times 1$      &  $512$             &$3 \times 2 \times 2 \times 1$    &  $1024$ \\   \hline
$8192 $     & $ 16 $           & $4 \times 2 \times 2 \times 1$      &  $1024$            &$2 \times 2 \times 2 \times 2$    &  $2048$ \\   \hline
$12288$     & $ 24 $           & $6 \times 2 \times 2 \times 1$      &  $1024$            &$3 \times 2 \times 2 \times 2$    &  $2048$ \\   \hline
\end{tabular}%
}
\end{table}

%% file: experiments.tex
\section{Experiments}\label{section:evaluation}
We support our theoretical predictions with the following experiments:
\begin{inparaenum}[(A)]
    \item bisection pairing\label{experiment-ping-pong};
    \item matrix multiplication\label{experiment-algo}; and
    \item simulation of strong scaling test\label{experiment-juqueen-spikes}.
\end{inparaenum}

The proposed partitions on Mira were made available for our experiments by the generous assistance of the system operators of Mira, who let us use a temporarily modified processor allocation policy. This did not require modifying the network's physical structure, but rather only the software-defined policy.

\subsection{Bisection pairing experiment}
\paragraph*{Experimental settings}
We performed a ping-pong benchmark using the \emph{furthest-node} scheme outlined in~\cite{underthehood}, which pairs nodes that are located at a maximal number of hops from each other.
The benchmark was performed as follows: each pair of compute nodes simultaneously sends to and receives from its counterpart a message of fixed size.
This was repeated for $30$ rounds without synchronization across distinct pairs.
The first $4$ rounds were treated as warm-up, and were not counted in the total time.
To prevent unexpected behaviors due to caching effects, the messages were randomly generated between each round.
A single link in the Blue~Gene/Q network topology has a bandwidth of $2$~Gigabyte per second per direction~\cite{underthehood}, and so the total communication volume between each pair of ranks was set to \(2\) Gigabytes, broken into \(16\) chunks sized \(0.1342\) Gigabytes each, to maximize the induced contention and its visibility.
We compared the performance of currently-used partitions against the proposed partitions on Mira, and measured the average time required for a pair of nodes to complete all rounds.
This was replicated on JUQUEEN, where we compared best-case and worst-case partitions.

\begin{figure}
	\includegraphics{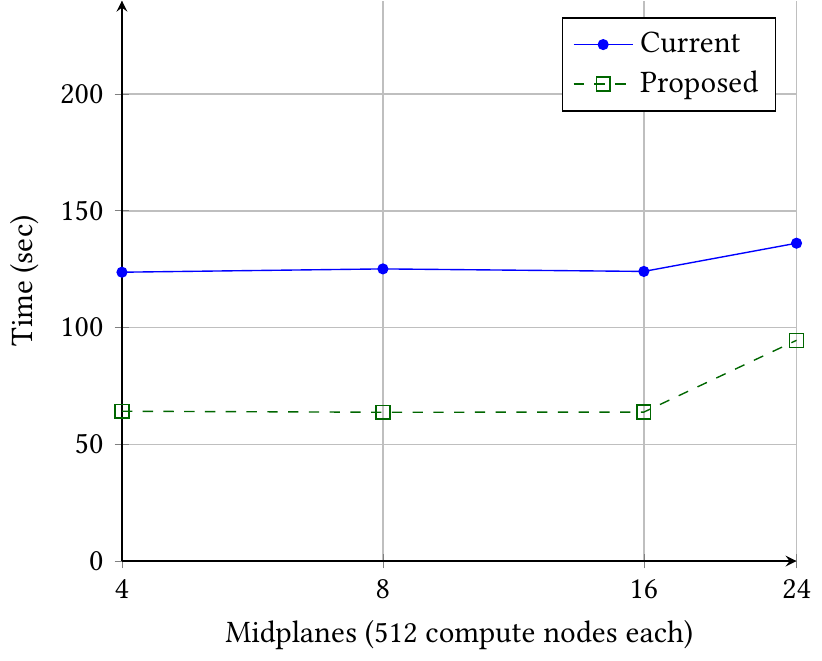}
	\caption{Mira: Bisection pairing experiment, using 4 warm-up rounds and 26 communication rounds, with messages of size $0.1342$ Gigabyte.}\label{plt-mira-pingpong}
\end{figure}

\begin{figure}
	\includegraphics{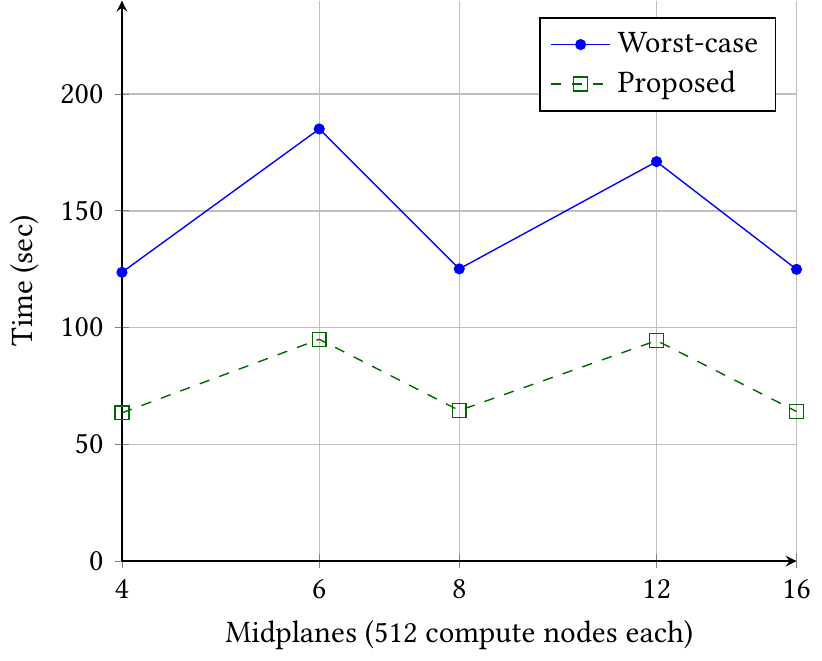}
	\caption{JUQUEEN: Bisection pairing experiment, using 4 warm-up rounds and 26 communication rounds, with messages of size $0.1342$ Gigabyte. As shown in Table 2, the average bisection bandwidth per node is identical for the \(4\) and \(8\) midplane partitions, but is \(50\%\) smaller for the \(6\) midplane partition. This is consistent with the observed results.}\label{plt-juqueen-pingpong}
\end{figure}

\paragraph*{Results}
The results for Mira and JUQUEEN are described in Figure~\ref{plt-mira-pingpong} and Figure~\ref{plt-juqueen-pingpong}, respectively.
On both Mira and JUQUEEN, the difference in average execution time is at least a factor of $1.92$ where the predicted factor is $2.00$ (except for $24$ midplanes on Mira, where it is $1.44$ and $1.50$, respectively).
This confirms the impact of partition geometry on the network contention and execution time, and shows our predicted speedup is attainable for contention-bound workloads on both systems.
An unexpected difference can be seen between the currently-defined and proposed partitions on Mira when running on $16$ and $24$ midplanes.
The $9.7\%$ increase between the currently-defined $16$ and $24$ partitions may be attributed to a combination of noise and low path diversity relative to the other partition geometries worsening the contention. In addition, the fact that some of the network links of the size $3$ dimension in the $24$ midplane partition are only utilized in one direction may have also caused a mild increase in effective resource contention.
For the proposed partitions, the increase between $16$ and $24$ partitions is expected, since the node count was increased by a factor of $1.5$ while the bisection bandwidth remained constant.
In summary, the bisection pairing experiment results agree almost perfectly with our predictions.

\subsection{Matrix multiplication experiment}
\paragraph*{Experimental settings}\label{par:matrix-limitations}
In order to measure the impact of our findings on real-life applications, we benchmark the performance of the Strassen-Winograd matrix multiplication algorithm.
We used the same set of parameters for equal-sized partitions. However, parameters were adjusted based on the partition sizes.
Table~\ref{tab:mira-matrix-params} details the parameters of each execution.
\input{tables/mira-matrix-parameters}
We used a parallel implementation by~\cite{caps,caps-impl} on random inputs between the different partition geometries.
The experimental constraints of that~\cite{caps,caps-impl} hold here.
Namely, there must be exactly $f\cdot7^k$ MPI ranks, where $f$ and $k$ are integers and $1 \le f \le 6$, and the matrix dimension must be a multiple of $f \cdot 2^r \cdot 7^{\left\lceil\frac{k}{2}\right\rceil}$.
We could not disable some of the compute nodes in a partition, as the additional network resources belonging to those `disabled' nodes would still be utilized by the system.
However, Mira has no partitions that contain exactly $7^k$ midplanes for any $k > 1$.
We therefore used multiple cores in each processor in order to create the required rank count, and tried to minimize the imbalance in compute and communication costs between the physical processors.
Parameter selection is described in Table~\ref{tab:mira-matrix-params}. For example, the execution on \(8\) midplanes had a total of \(31,213\) MPI ranks, and each compute node was allowed to use up to 8 cores (where each core may only be associated with a single rank).

\paragraph*{Results}
Time spent performing computation does not significantly differ between partition geometries of the same size.
These computation costs are $0.554, 0.5115, 0.4965$ and $0.0604$ seconds for $4, 8, 16$ and $24$ midplanes, respectively.
Communication costs of runs using proposed partitions were smaller by factors of $\times 1.37$ up to $\times 1.52$ than all executions which utilized currently-defined partitions (see Figure~\ref{plt-mira-matrix}).
The total wall-clock time was smaller by factors of $\times 1.08$ up to $\times 1.22$, due to the common computation costs.
\begin{figure}[h]
	\includegraphics{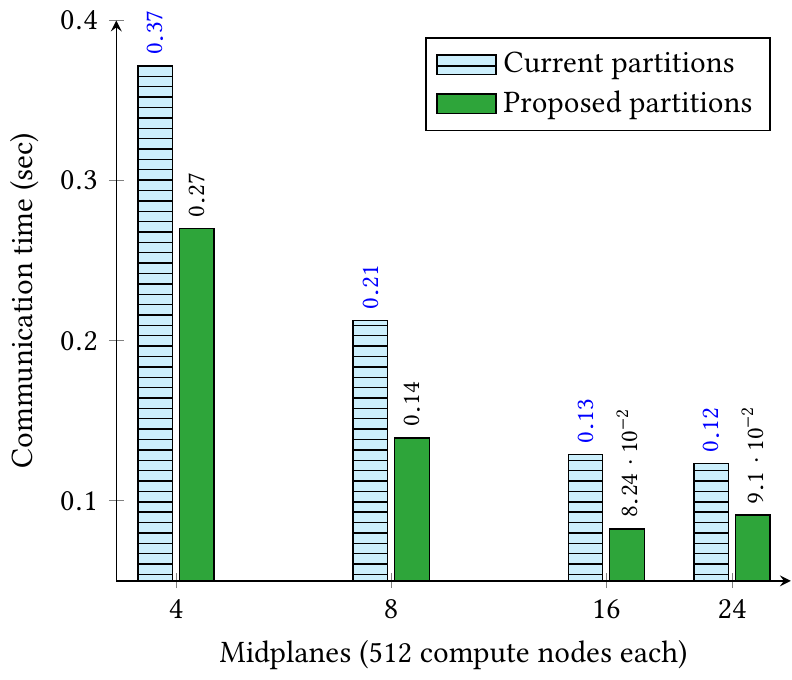}
	\caption{Mira: Matrix multiplication experiment. Three executions were performed for each midplane count and each partition type.
		Costs offset by communication-hiding are not presented. For both partition types, they are $0.059, 0.067, 0.099$ and $0$ seconds for $4,8,16$ and $24$ midplanes, respectively. Results are described using communication time instead of wallclock as the additional computation time for identical workloads is not relevant to the contention costs.
	}\label{plt-mira-matrix}
\end{figure}

\subsection{Simulation of strong scaling test}
\paragraph*{Motivation}
When optimal and sub-optimal partitions are randomly selected by the job scheduler, unnoticed variations in the bisection bandwidth can potentially result in false conclusions regarding the scaling behavior of an algorithm.
The purpose of this experiment is to test the possibility of contention costs to cause a parallel algorithm that has good strong scaling properties to appear as though it does not.
For example, consider the results of the bisection pairing experiment on JUQUEEN (see Figure~\ref{plt-juqueen-pingpong}) and suppose the runs on up to $6$ midplanes are assigned only proposed partitions, but the runs on $8$ and $12$ midplanes are assigned only worst-case partitions.
Without knowledge of the bisection bandwidths available to each execution, the runtime may seem to increase linearly with midplane count, which is clearly incorrect.

\paragraph*{Experimental settings}
We used the same code as in Experiment~\ref{experiment-algo}, and were subject to the same constraints in parameter selection.
We could not use more than $3$ distinct midplane counts in the experiment without altering other parameters such as the matrix dimension.
\input{tables/mira-scaling}
In order to show two scalability plots with a common point, we chose midplane sizes $2, 4$ and $8$, and a matrix dimension of $9408$ storing double-precision values.
There is only one way to define a cuboid of $2$ midplanes, and so the smallest execution used a partition common to both the current and proposed geometries.

\begin{figure}[h]
	\includegraphics{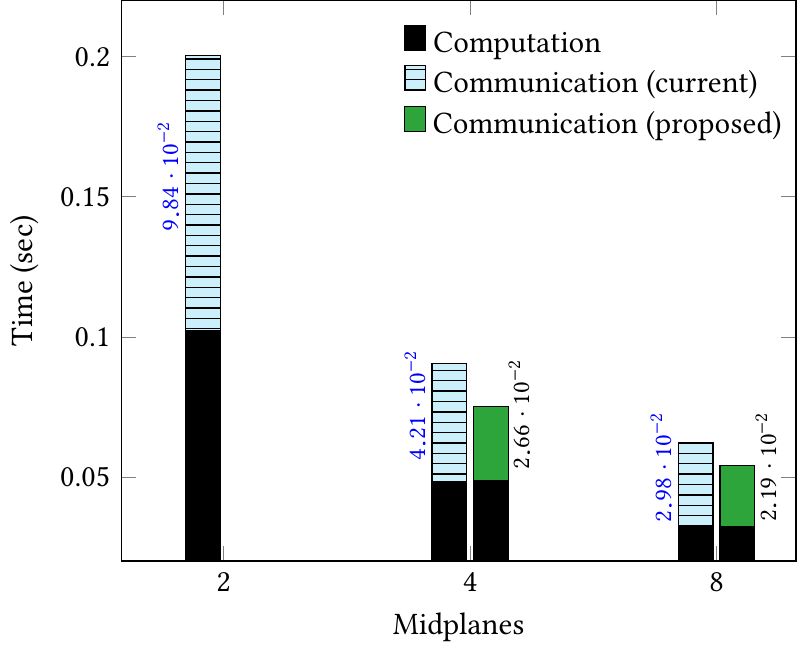}
	\caption{Mira: Strong-scaling experiment. The run on $2$ midplanes allows only one partition geometry. Floating values indicate communication times. Communication costs hidden by computations were not counted.} \label{plt-mira-scaling}
\end{figure}

\paragraph*{Results}
This experiment was partly successful, as the results indicate a linear decrease of communication costs when strong scaling from $2$ midplanes to $8$ using the proposed partition geometry, but only a sub-linear decrease when using the current partitions.
Therefore, a user should be advised that a test of the computation's strong scaling behavior using the currently-defined partitions may incorrectly indicate it cannot linearly scale beyond $4$ midplanes.

Runs on all partition types exhibited super-linear scaling of the communication costs between $2$ and $4$ midplanes, regardless of partition geometry used. As the difference in bisection bandwidth is exactly $\times 2$, link contention alone cannot account for the effect.

The exact source of the super-linear speedup is not fully understood. It may be related to the fact that the data fits entirely within the shared L2 cache when using \(4\) and \(8\) midplanes, but not when using \(2\) midplanes.
Perhaps this allowed the dedicated communications core to more efficiently move the data.

More specifically, each of Mira's processors has $32$ Megabyte of L2 cache storage shared by all its cores, and an additional core used for communications. This means $32, 64,$ and $128$ Gigabyte of combined L2 storage for $2,4,$ and $8$ midplanes, respectively.
The execution pattern the BFS-DFS matrix multiplication algorithm used was $4$ BFS steps, requiring a combined minimum of $3 \cdot \left(\frac{7}{4}\right)^4 \cdot 8 \cdot 9408^2$ bytes, or $18.55$ Gigabyte in order to store all matrices across all processors, added with a similar amount of space for the communications library buffers.
When the overall space requirement exceeds the L2 memory of $2$ midplanes, this results in cache misses and use of the slower RAM, hence a slowdown for the executions on $2$ midplanes.
This somewhat muddles the visibility of scalability properties of the fast matrix multiplication computation.
It remains evident that the scaling is better when the proposed geometries are used.
Specifically, the computation on $2$ midplanes exhibits a $\times 4.4$ decrease in communication costs on $8$ midplanes in a proposed geometry, and a $\times 3.3$ decrease when using the current geometry.
Therefore, the fast matrix multiplication algorithm may be incorrectly surmised to have a smaller strong scaling range on Mira if evaluated only using the current partition geometries.
More extreme disparities are possible: given the bisection bandwidth of the $2$ and $4$ midplane partitions in Table~\ref{tab:mira-scaling-params}, a computation's wallclock time may remain identical on both $2$ and $4$ midplanes even if the computation can linearly scale to $16$ midplanes and beyond.

%% file: tables/mira-matrix-parameters.tex
\begin{table}
\caption{Parameters of the matrix multiplication experiment on Mira.}\label{tab:mira-matrix-params}
\resizebox{\columnwidth}{!}{%
\begin{tabular}{|c|c||c|c|c|c|} \hline
\thead{$P$} & \thead{Midplanes} & \thead{MPI Ranks} & \thead{Max. active cores} & \thead{Avg. cores per proc}      & \thead{Matrix dimension} \\   \hline
$2048 $     & $ 4 $             & $31213$           &  $16$                        & $15.24$                          & $32928$                  \\   \hline
$4096 $     & $ 8 $             & $31213$           &  $8$                         & $7.62$                           & $32928$                   \\   \hline
$8192 $     & $ 16 $            & $31213$           &  $4$                         & $3.81$                           & $32928$                   \\   \hline
$12288$     & $ 24 $            & $117649$          &  $16$                        & $9.57$                           & $21952$                   \\   \hline
\end{tabular}%
}
\end{table}

%% file: tables/mira-scaling.tex
\begin{table}[h]
\caption{Strong scaling experiment parameters on Mira, performing matrix multiplication with dimension $9408$.}\label{tab:mira-scaling-params}
\resizebox{\columnwidth}{!}{%
\begin{tabular}{|l|c||c|c|c|c|c|} \hline
\thead{$P$} & \thead{Midplanes} & \thead{MPI Ranks} & \thead{Max. active cores}    & \thead{Avg. cores per proc} & \thead{Current Bw} & \thead{Proposed BW}\\   \hline
$1024 $     & $ 2 $             & $2401$            &  $4$                         & $2.34$                      & $256$ & $256$ \\   \hline
$2048 $     & $ 4 $             & $4802$            &  $4$                         & $2.34$                      & $256$ & $512$  \\   \hline
$4096 $     & $ 8 $             & $9604$            &  $4$                         & $2.34$                      & $512$ & $1024$   \\   \hline
\end{tabular}%
}
\end{table}

%% file: discussion.tex
\section{Discussion}\label{section:discussion}
In this work, we focus on potential performance boosts due to improved internal bisection bandwidth of partitions.
Determining the importance of such speedups relative to other possible machine design optimization goals is beyond the scope of this work.
Particularly, there are many motivations to the design and installation of specific supercomputer systems, as well as to setting a processor allocation policy. 
Such motivations may include computational kernels or specific job sizes deemed particularly important for the system, or even specific software that is intended to be executed often.
Further reasons may include packing of jobs affecting overall system utilization, cabling complexity, cooling, and ease of access and maintenance.

\input{tech-transfer}

\paragraph*{Sequoia}\label{sub:sequoia}
Installed at Lawrence Livermore National Laboratory, it is the largest Blue~Gene/Q system in production.
Sequoia is ranked $6$th in the November 2017 list of top 500 supercomputers~\cite{TOP500-nov17}.
It has $98304$ compute nodes, with network size $16 \times 16 \times 16 \times 12 \times 2$, or $4 \times 4 \times 4 \times 3$ midplanes~\cite{sequoia-network}.
Sequoia's scheduler seems to support all partition geometries supportable that the Blue~Gene/Q network allows (similarly to JUQUEEN).
Hence, both optimal and sub-optimal permissible partitions may be defined for certain midplane counts.
Sequoia transitioned into classified work in $2013$~\cite{sequoia-transition}.
We thus could not perform experiments on that system, but depending on its allocation policy it may be possible to improve its network performance using our analysis.

\subsection*{Machine design}
The ratio of maximal dimension and total machine sizes influences the global bisection bandwidth, we can reason about the design of an entire Blue~Gene/Q network.
Recall that JUQUEEN has a network size of $7 \times 2 \times 2 \times 2$ ($56$ midplanes in total).
We consider similar machines with $48$ and $54$ midplanes (denoted here JUQUEEN-48 and JUQUEEN-54, respectively), with more balanced dimension sizes.
JUQUEEN-54 has dimensions $3 \times 3 \times 3 \times 2$, and JUQUEEN-48 has dimensions $ 4 \times 3 \times 2 \times 2$.

Mira has network size $4 \times 4 \times 3 \times 2$. As the networks of JUQUEEN-54 and JUQUEEN-48 are both subgraphs of Mira's, their physical construction is clearly feasible.

Both these machines have fewer midplanes than JUQUEEN, but have better (greater) bisection bandwidth due to their network sizes.
\begin{figure}
	\resizebox{\columnwidth}{!}{
	\includegraphics{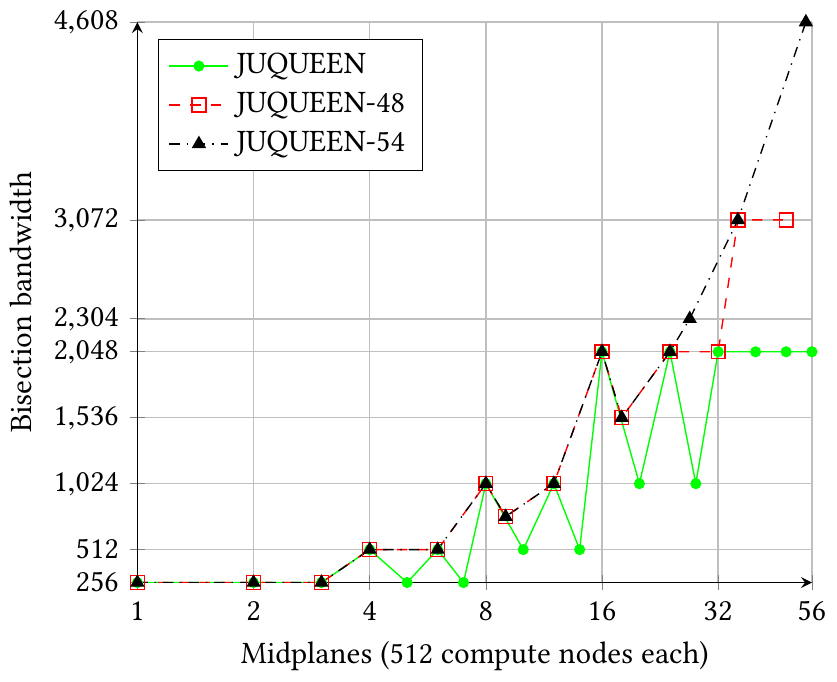}
	}
	\caption{Normalized bisection bandwidth comparison between JUQUEEN and hypothetical machines JUQUEEN-$48$ and JUQUEEN-$54$. We assume that JUQUEEN always uses best-case partitions.}\label{juqueen-partitions-theoretical}
\end{figure}
A comparison of bisection bandwidths of partitions of those theoretical machines against the optimal allocations of JUQUEEN is presented in
Figure~\ref{juqueen-partitions-theoretical}, and a full listing of optimal partitions for all three machines appears in Table~\ref{tbl-juqueen-machines}.
The network bandwidths of partitions of both theoretical machines are identical to those of JUQUEEN when utilizing smaller partitions, and strictly greater on the largest partition sizes.
Since JUQUEEN is a larger system, applications that are able to perfectly strong scale to its full size will still exhibit superior performance when executed on the entire machine.
However, in nearly all other cases, JUQUEEN cannot outperform the smaller theoretical machines.
On contention-bound workloads, the suggested machines are predicted to perform at least as well as JUQUEEN, and attain speedup factors up to~$\times 2$ and~$\times 1.5$ for JUQUEEN-54 and JUQUEEN-48, respectively.
The partition geometries of proposed machines are described in Table~\ref{tbl-juqueen-machines}.
\input{tables/juqueen-machines}
\paragraph*{Future Work}
Our conjecture about the optimality of Equation~\ref{general-torus-sse} for arbitrary subsets remains open.

We believe further speedups on Blue~Gene/Q can be demonstrated for several kernels of interest.
Direct $N$-body simulation have greater asymptotic contention cost lower bounds than fast matrix multiplication~\cite{contention}, increasing the impact of the internal bisection bandwidth.
High-performance implementations of FFT, classical matrix multiplication, and other common kernels may better utilize the available hardware resources, decreasing the ratio of time spent performing computation.
For both those cases, the impact of internal bisection bandwidth on wallclock time is predicted to be greater than in Experiment~\ref{experiment-algo}.

Similar isoperimetric analysis can be conducted on other networks to potentially improve processor allocation policies, and to ensure contention-related effects do not unnecessarily inhibit scaling.

Testing bisection sensitivity of machine benchmarks can be done by comparing the score of equal-sized partitions with different bisection bandwidths.

Designing new network topologies, and evaluating existing ones, should be done with their partitioning constraints and internal bisection bandwidths in mind.
Such considerations can reveal specific partition sizes for which the network performs poorly, and makes it easier to solve such issues.

Processor allocation policy decisions of job schedulers can be improved if they are informed whether a given computation is expected to be network-bound or not.
For example, if a partition with sub-optimal bisection bandwidth is currently available for use, a scheduler may decide whether to allocate it to a pending job, or to wait for a partition with better bisection bandwidth.
This decision can be contingent on a user-provided hint which indicates whether the job is expected to be contention-bound or not.

\input{conclusions}

%% file: tech-transfer.tex
\subsection*{Application to other topologies}\label{sub:tech-transfer}
We discussed the IBM Blue Gene/Q network topology in great detail, but
our method applies to arbitrary network topologies if edge-isoperimetric problems can be efficiently solved on their network graphs.
When the network graph is regular and has uniform link capacity -- which is the case in almost all networks of supercomputers (except Dragonfly; see below) -- isoperimetric analysis is sufficient to determine the small-set expansion of the graph. This provides additional information to merely the bisection bandwidth, and can predict contention bottlenecks at locations other than the network bisection~\cite{contention}.

The ToFu interconnect used by the K~Computer~\cite{ajima2012tofu} is a high-dimensional torus with certain similarities to Blue~Gene/Q.
Torus networks of lower dimension, such as the Cray XK7 $3D$-torus machine Titan~\cite{supercomp-titan}, may require a formulation of the edge-isoperimetric problem that considers weighted edges.

For hypercube-based supercomputers such as Pleiades~\cite{supercomp-pleiades}, the edge-isoperimetric problem is long solved in~\cite{harper1964optimal}, and so our method is directly usable.

For Fat-Tree topologies, the application of our method is more challenging. If the processor allocation policy permits distinct jobs to share network resources, then the available link capacity may be smaller than isoperimetric analysis alone would indicate.
If sharing of network resources is forbidden, then the policy is expected to be so constrained that our method will not be able to obtain improvements.

HyperX networks are Cartesian products of cliques $K_{a_1},\ldots,K_{a_D}$.
The number of cliques in the product and their exact sizes are both variable. Each clique may have a different link capacity; when all links have the same capacity, the HyperX network is said to be \emph{regular}.
Finding an optimal HyperX structure for a fixed vertex count is performed by exhaustive search~\cite{hyperx}.
The network bisection bandwidth is attained by selecting half of the vertices in $K_i$ for some $1 \le i \le D$ and all vertices in $K_{j}$ for $i \neq j$~\cite{hyperx}.
The edge-isoperimetric problem for regular HyperX network graphs is solved in~\cite{lindsey1964assignment}, by choosing vertices of the product cliques in order of descending size. 

Dragonfly networks~\cite{kim2008technology} as implemented in the Cray XC series~\cite{faanes2012cray} are a collection of `groups' each containing up to $96$ Aries routers, where each group is an instance of $K_{16} \times K_6$.  
Links belonging to the $K_6$ clique have a normalized capacity of $3$ relative to the $K_{16}$ links, requiring a weighted version of the edge-isoperimetric problem to be used.
Unlike HyperX, the Dragonfly network also contains inter-group links with a normalized capacity of $4$.
To apply our method to a Dragonfly-based system, it is necessary to model the inter-group links to create the network graph.
We are unaware of any public description of the inter-group link arrangement, but~\cite{dragonfly-global} discusses three possible schemes for such systems.
A further minor challenge is the pairing of Aries routers: unlike an edge in a simple graph, each endpoint of an inter-group link is a pair of adjacent Aries routers.
It may thus be necessary to introduce the constraint of $t$ even when considering edge-isoperimetric problems on these networks.

The Slim Fly network topology is more difficult to analyze in the general case, since the cabling layout varies greatly based on the global network size, necessitating exhaustive search~\cite{besta2014slim}.
Given the complexity of finding such constructions, the existence of a general solution to edge-isoperimetric problems that fits all possible constructions seems unlikely.

%% file: tables/juqueen-machines.tex
\begin{table}[ht]\caption{Full list of best-case partitions in JUQUEEN and the two proposed machines JUQUEEN-54 and JUQUEEN-48. Dimensions are listed in sorted order, BW is normalized bisection bandwidth.}\label{tbl-juqueen-machines}
\centering
\resizebox{\columnwidth}{!}{%
        \begin{tabular}{|l|l||c|c||c|c||c|c|} \hline
\thead{$P$}& \thead{Midplanes} & \thead{JUQUEEN}      & \thead{J BW} & \thead{JUQUEEN-54} & \thead{J-54 BW} & \thead{JUQUEEN-48} & \thead{J-48 BW}\\   \hline
$512  $    &   $1$      &$1 \times 1\times 1\times 1$ &  $256$   &$1 \times 1\times 1\times 1$& $256$& $1 \times 1\times 1\times 1$& $256$ \\   \hline
$1024 $    &   $2$      &$2 \times 1\times 1\times 1$ &  $256$   &$2 \times 1\times 1\times 1$& $256$& $2 \times 1\times 1\times 1$& $256$ \\   \hline    
$1536 $    &   $3$      &$3 \times 1\times 1\times 1$ &  $256$   &$3 \times 1\times 1\times 1$& $256$& $3 \times 1\times 1\times 1$& $256$ \\   \hline
$2048 $    &   $4$      &$2 \times 2\times 1\times 1$ &  $512$   &$2 \times 2\times 1\times 1$& $512$& $2 \times 2\times 1\times 1$& $512$ \\   \hline
$2560 $    &   $5$      &$5 \times 1\times 1\times 1$ &  $256$   &                            &   &                                &       \\   \hline
$3072 $    &   $6$      &$3 \times 2\times 1\times 1$ &  $512$   &$3 \times 2\times 1\times 1$& $512$& $3 \times 2\times 1\times 1$& $512$ \\   \hline
$3584 $    &   $7$      &$7 \times 1\times 1\times 1$ &  $256$   &                            &   &                                &  \\   \hline
$4096 $    &   $8$      &$2 \times 2\times 2\times 1$ &  $1024$  &$2 \times 2\times 2\times 1$& $1024$& $2 \times 2\times 2\times 1$& $1024$ \\   \hline
$4608 $    &   $9$      &                             &          &$3 \times 3\times 1\times 1$& $768$& $3 \times 3\times 1\times 1$& $768$ \\   \hline
$5120 $    &   $10$     &$5 \times 2\times 1\times 1$ &  $512$   &                            &   &                                &  \\   \hline
$6144 $    &   $12$     &$3 \times 2\times 2\times 1$ &  $1024$  &$3 \times 2\times 2\times 1$& $1024$& $3 \times 2\times 2\times 1$& $1024$ \\   \hline
$7168 $    &   $14$     &$7 \times 2\times 1\times 1$ &  $512$   &                            &   &                             & \\   \hline
$8192 $    &   $16$     &$2 \times 2\times 2\times 2$ &  $2048$  &$2 \times 2\times 2\times 2$& $2048$& $2 \times 2\times 2\times 2$& $2048$ \\   \hline
$9216 $    &   $18$     &                             &          &$3 \times 3\times 2\times 1$& $1536$& $3 \times 3\times 2\times 1$& $1536$ \\   \hline
$10240$    &   $20$     &$5 \times 2\times 2\times 1$ &  $1024$  &                            &   &                             & \\   \hline
$12288$    &   $24$     &$3 \times 2\times 2\times 2$ &  $2048$  &$3 \times 2\times 2\times 2$& $2048$& $3 \times 2\times 2\times 2$& $2048$ \\   \hline
$13824$    &   $27$     &                             &          &$3 \times 3\times 3\times 1$& $2304$&                             & \\   \hline
$14336$    &   $28$     &$7 \times 2\times 2\times 1$ &  $1024$  &                            &   &                             & \\   \hline
$16384$    &   $32$     &$4 \times 2\times 2\times 2$ &  $2048$  &                            &   & $4 \times 2\times 2\times 2$& $2048$ \\   \hline        
$18432$    &   $36$     &                             &          &$3 \times 3\times 2\times 2$& $3072$& $3 \times 3\times 2\times 2$& $3072$ \\   \hline
$20480$    &   $40$     &$5 \times 2\times 2\times 2$ &  $2048$  &                            & &                             & \\   \hline
$24576$    &   $48$     &$6 \times 2\times 2\times 2$ &  $2048$  &                            & & $4 \times 3\times 2\times 2$& $3072$ \\   \hline            
$27648$    &   $54$     &                             &          &$3 \times 3\times 3\times 2$& $4608$&                             & \\   \hline
$28672$    &   $56$     &$7 \times 2\times 2\times 2$ &  $2048$  &                            & &                             &  \\   \hline
        \end{tabular}%
    }
\end{table}

%% file: conclusions.tex
\subsection*{Conclusions}
We presented a method for analyzing processor allocation policies using an isoperimetric analysis of the network graph, and determining whether any partition geometries induce sub-optimal internal bisection bandwidth.
We applied our method to two leading Blue~Gene/Q supercomputers; demonstrated performance improvements for various workloads; and have shown how to apply our method to other networks.

%% file: acknowledgments.tex
\section*{Acknowledgment}
We thank Ivo Kabadshow and Dorian Krause of J{\"u}lich Supercomputing Centre for their help in arranging the JUQUEEN experiments.
We thank Adam Scovel of Argonne National Laboratory for his help and support in setting up custom partitions on Mira.
Our experiments could not have been done without their help.

The authors gratefully acknowledge the Gauss Centre for Supercomputing e.V. (www.gauss-centre.eu) for funding this project by providing computing time through the John von Neumann Institute for Computing (NIC) on the GCS Supercomputer JUQUEEN at J{\"u}lich Supercomputing Centre (JSC). 
This research used resources of the Argonne Leadership Computing Facility, which is a DOE Office of Science User Facility supported under Contract DE-AC02-06CH11357.

Research is supported by grants 1878/14, and 1901/14 from the Israel Science Foundation (founded by the Israel Academy of Sciences and Humanities) and grant 3-10891 from the Ministry of Science and Technology, Israel.
Research is also supported by the Einstein Foundation and the Minerva Foundation.
This work was supported by the PetaCloud industry-academia consortium.
This research was supported by a grant from the United States-Israel Bi-national Science Foundation (BSF), Jerusalem, Israel.
This project has received funding from the European Research Council (ERC) under the European Union's Horizon 2020 research and innovation programme (grant agreement No 818252). This work was supported by The Federmann Cyber Security Center in conjunction with the Israel national cyber directorate. 

%% file: appendix.tex
\appendix

\section{Machine Partitions}
\input{tables/mira-partition-improvements-table.tex}
\input{tables/juqueen-table-best-worst.tex}

%% file: tables/mira-partition-improvements-table.tex
\begin{table}[h]
\caption{Mira: normalized bisection bandwidths of all current and proposed partitions.}\label{tbl-mira-parts-all}
\resizebox{\columnwidth}{!}{
    \begin{tabular}{|l|c||l|c||l|c|} \hline
    \thead{$P$} &  \thead{Midplanes} & \thead{Current Geometry}         & \thead{BW}  & \thead{New Geometry}            & \thead{New BW}\\ \hline
    $512$       &  $ 1 $             & $1 \times 1 \times 1 \times 1 $  & $256$       &                                 &               \\ \hline
    $1024$      &  $ 2 $             & $2 \times 1 \times 1 \times 1 $  & $256$       &                                 &               \\ \hline
    $2048$      &  $ 4 $             & $4 \times 1 \times 1 \times 1 $  & $256$       & $2 \times 2 \times 1 \times 1$  & $512$         \\ \hline
    $4096$      &  $ 8 $             & $4 \times 2 \times 1 \times 1 $  & $512$       & $2 \times 2 \times 2 \times 1$  & $1024$        \\ \hline
    $8192$      &  $ 16 $            & $4 \times 4 \times 1 \times 1 $  & $1024$      & $2 \times 2 \times 2 \times 2$  & $2048$        \\ \hline
    $12288$     &  $ 24 $            & $4 \times 3 \times 2 \times 1 $  & $1536$      & $3 \times 2 \times 2 \times 2$  & $2048$        \\ \hline
    $16384$     &  $ 32 $            & $4 \times 4 \times 2 \times 1 $  & $2048$      &                                 &               \\ \hline
    $24576$     &  $ 48 $            & $4 \times 4 \times 3 \times 1 $  & $3072$      &                                 &               \\ \hline
    $32768$     &  $ 64 $            & $4 \times 4 \times 2 \times 2 $  & $4096$      &                                 &               \\ \hline
    $49152$     &  $ 96 $            & $4 \times 4 \times 3 \times 2 $  & $6144$      &                                 &               \\ \hline
    \end{tabular}
}
\end{table}

%% file: tables/juqueen-table-best-worst.tex
\begin{table}[h]\caption{Full list of JUQUEEN allocation best and worst cases by compute node count $P$. Dimensions are listed in sorted order, BW is bisection bandwidths normalized by link capacity.}\label{tab:juqueen-full-allocations}
\centering
\resizebox{\columnwidth}{!}{%
        \begin{tabular}{|l|c||c|c||c|c|} \hline
            \thead{$P$}& \thead{Midplanes} & \thead{Worst-case Geometry}      & \thead{Worst BW}  & \thead{Proposed Geometry} & \thead{Proposed BW}    \\   \hline
            $512  $    &   $1$      & $1\times 1\times 1\times 1 $    &  $256$             &                                  &         \\   \hline
            $1024 $    &   $2$      & $2\times 1\times 1\times 1 $    &  $256$             &                                  &         \\   \hline    
            $1536 $    &   $3$      & $3\times 1\times 1\times 1 $    &  $256$             &                                  &         \\   \hline
            $2048 $    &   $4$      & $4\times 1\times 1\times 1 $    &  $256$             &$2 \times 2\times 1\times 1$      &  $512$  \\   \hline
            $2560 $    &   $5$      & $5\times 1\times 1\times 1 $    &  $256$             &                                  &         \\   \hline
            $3072 $    &   $6$      & $6\times 1\times 1\times 1 $    &  $256$             &$3 \times 2\times 1\times 1$      &  $512$  \\   \hline
            $3584 $    &   $7$      & $7\times 1\times 1\times 1 $    &  $256$             &                                  &         \\   \hline
            $4096 $    &   $8$      & $4\times 2\times 1\times 1 $    &  $512$             &$2 \times 2\times 2\times 1$      &  $1024$ \\   \hline
            $5120 $    &   $10$     & $5\times 2\times 1\times 1 $    &  $512$             &                                  &         \\   \hline
            $6144 $    &   $12$     & $6\times 2\times 1\times 1 $    &  $512$             &$3 \times 2\times 2\times 1$      &  $1024$ \\   \hline
            $7168 $    &   $14$     & $7\times 2\times 1\times 1 $    &  $512$             &                                  &         \\   \hline
            $8192 $    &   $16$     & $4\times 2\times 2\times 1 $    &  $1024$            &$2 \times 2\times 2\times 2$      &  $2048$ \\   \hline
            $10240$    &   $20$     & $5\times 2\times 2\times 1 $    &  $1024$            &                                  &         \\   \hline
            $12288$    &   $24$     & $6\times 2\times 2\times 1 $    &  $1024$            &$3 \times 2\times 2\times 2$      &  $2048$ \\   \hline
            $14336$    &   $28$     & $7\times 2\times 2\times 1 $    &  $1024$            &                                  &         \\   \hline
            $16384$    &   $32$     & $4\times 2\times 2\times 2 $    &  $2048$            &                                  &         \\   \hline        
            $20480$    &   $40$     & $5\times 2\times 2\times 2 $    &  $2048$            &                                  &         \\   \hline
            $24576$    &   $48$     & $6\times 2\times 2\times 2 $    &  $2048$            &                                  &         \\   \hline            
            $28672$    &   $56$     & $7\times 2\times 2\times 2 $    &  $2048$            &                                  &         \\   \hline
        \end{tabular}%
    }

\end{table}